\documentclass[11pt]{article}

\usepackage{amsfonts, amssymb, amsmath, amsthm, mathrsfs, chet}
\newcommand{\bra}[1]{{\langle#1|}}
\newcommand{\ket}[1]{{|#1\rangle}}

\theoremstyle{plain}
\newtheorem{thm}{Theorem}
\newtheorem{lem}[thm]{Lemma}

\newtheorem{cor}[thm]{Corollary}

\theoremstyle{definition}
\newtheorem{defn}{Definition}

\preprint{MIT-CTP-4751}

\title{\Large Performance of QAOA on Typical Instances of \\\vspace{10pt} Constraint Satisfaction Problems with Bounded Degree}

\author{Cedric Yen-Yu Lin$^{a,b}$ and Yechao Zhu$ ^{b} $ \emails{cedricl@umiacs.umd.edu, eltonzhu@mit.edu}}

\affiliation{$ ^a $Joint Center For Quantum Information and Computer Science,\\\vspace{-2pt}University of Maryland, College Park, Maryland 20742, USA\\
$ ^b $Center For Theoretical Physics and Department of Physics, \\\vspace{-2pt}Massachusetts Institute of Technology, Cambridge, Massachusetts 02139, USA}

\abstract{We consider constraint satisfaction problems of bounded degree, with a good notion of "typicality", e.g. the negation of the variables in each constraint is taken independently at random. Using the quantum approximate optimization algorithm (QAOA), we show that $ \mu+\Omega(1/\sqrt{D}) $ fraction of the constraints can be satisfied for typical instances, with the assignment efficiently produced by QAOA. We do so by showing that the averaged fraction of constraints being satisfied is $ \mu+\Omega(1/\sqrt{D}) $, with small variance. Here $ \mu $ is the fraction that would be satisfied by a uniformly random assignment, and $ D $ is the number of constraints that each variable can appear. CSPs with typicality include Max-$ k $XOR and Max-$ k $SAT. We point out how it can be applied to determine the typical ground-state energy of some local Hamiltonians. We also give a similar result for instances with "no overlapping constraints", using the quantum algorithm. We sketch how the classical algorithm might achieve some partial result.}

\begin{document}
\maketitle
\section{Introduction}
A constraint satisfaction problem seeks to find an assignment that satisfies a maximum number of constraints, which are predicates over some variables.

From the algorithmic perspective, one is interested in finding an assignment that either satisfies all the constraints, or satisfies as many constraints as possible. For the latter problem, there is an important constant $ \mu $, which is the expected fraction of constraints satisfied by a random assignment. It can be shown, using the method of conditional expectation, that one can always find an assignment satisfying $ \mu $ fraction of the constraints. Then the question becomes whether one can find assignments which satisfy significantly more than that. In this aspect, H{\aa}stad has proved some remarkable inapproximability results\cite{Hastad01}. He showed that for Max-3SAT (where $ \mu=\frac{7}{8} $), given there exist assignments which satisfy all the constraints, it is NP-hard to find an assignment that satisfies $ \frac{7}{8}+\delta $ of them, for any $ \delta>0 $. There are similar results for Max-3XOR and other CSPs.

One can also consider CSPs with bounded degree, i.e each variable occurs in at most $ D $ constraints. Given such restrictions, H{\aa}stad showed these problems can be approximated within $ (\mu+\Omega(1/D))^{-1} $ \cite{Hastad00}. However, as shown by the case of Max-Cut on $ D $-cliques, this is the best possible for general CSPs.

Recently, Farhi, Goldstone, and Gutmann introduced the Quantum Approximate Optimization Algorithm (QAOA)\cite{Farhi1411}, and used it to give an efficient quantum algorithm which finds an assignment that satisfies $ \frac{1}{2}+\Omega\left(\frac{1}{\sqrt{D}\ln D}\right) $ fraction of the constraints, for Max-3XOR with bounded degree\cite{Farhi1412}. Later, Barak \textit{et al}. gave a classical algorithm which finds an assignment that satisfies $ \frac{1}{2}+\Omega(1/\sqrt{D}) $ fraction of the constraints, for Max-$ k $XOR with bounded degree, with $ k $ odd\cite{Barak15}. They further showed that this result is optimal. Despite these exciting breakthroughs, it is still unknown whether one can obtain similar results for Max-3SAT, and other CSPs in general.

Farhi \textit{et al}. also showed that, in the typical case, their quantum algorithm can output an assignment that satisfies $ \frac{1}{2}+\Omega(1/\sqrt{D}) $ fraction of the constraints, for Max-3XOR.

In this work, we formally define the notion of typicality for CSPs, and show that on average, one can efficiently output an assignment which satisfies $ \mu+\Omega(1/\sqrt{D}) $ fraction of the constraints, for such CSPs with bounded degree. CSPs with typicality include Max-$ k $XOR and Max-$ k $SAT. While our result is worse than Barak \textit{et al}.'s result for Max-$ k $XOR when $ k $ is odd, the other cases are rather interesting. This is achieved with a quantum algorithm (QAOA).

The paper is organized as follows. In the next section, we review the basic concepts in constraint satisfaction problems and QAOA, and formally define the notion of typicality. In section \ref{sec:MAX2XOR}, we use Max2XOR as an example to show how QAOA can give an advantage of $ \Omega(1/\sqrt{D}) $ for most instances. In section \ref{sec:typicality}, we show how QAOA can give an advantage of $ \Omega(1/\sqrt{D}) $ for most instances of CSPs with typicality, with small variance. In section \ref{sec:nooverlapping}, we show that QAOA can also give an advantage of $ \Omega(1/\sqrt{D}) $ for instances with "no overlapping constraints". In section \ref{sec:classical}, we sketch how the averaged advantage can be achieved with the classical algorithm developed by Barak \textit{et al.}, for some cases. Some details of the proofs are left in the appendices. 
\section{Preliminaries}\label{sec:preliminaries}
\begin{defn}[{\cite[Definition~7.22]{Odonnell14}}]
A constraint satisfaction problem (CSP) over domain $ \Omega $ is defined by a finite set of \textit{predicates} ("types of constraints") $ \Psi $, with each $ \psi\in\Psi $ being of the form $ \psi: \Omega^r~\to~\{0,1\} $ for some arity $ r $ (possibly different for different predicates).
\end{defn}

\begin{defn}[{\cite[Definition~7.24]{Odonnell14}}]
An instance $ \mathscr{P} $ of CSP($ \Psi $) over variable set $ V $ is a list of constraints. Each constraint $ C\in\mathscr{P} $ is a pair $ C=(\psi,\mathcal{S}) $ where $ \psi\in\Psi $ and where the scope $ \mathcal{S}=(x_1,\dots,x_r) $ is a tuple of distinct variables from $ V $.
\end{defn}

Given an assignment of the variables, one is either interested in the simultaneous satisfiability of all constraints, or the maximum number of constraints that can be satisfied. The latter problem is called Max-CSP. Typically we consider a CSP of $ m $ constraints and over $ n $ variables. Then the constraints are labeled as $ C_l=(\psi_l,\mathcal{S}_l) $ for $ l\in[m] $. Denote $ K_l $ the position of the tuples in $ \{x_1,\dots, x_n\} $. $ K_l\subset [n] $.

In this paper, we consider \textit{Boolean} CSPs of arity bounded by $ k $. By Boolean, we mean $ \left|\Omega\right|=2 $. For convenience, we choose $ \Omega=\{\pm1\} $. Since we've restricted the predicates to Boolean functions, the following theorem is particularly useful.

\begin{thm}[{\cite[Theorem~1.1]{Odonnell14}}]
Every function $ f:\{-1,1\}^k\to \mathbb{R} $ can be uniquely expressed as a multilinear polynomial,
\begin{equation}
f(x)=\sum_{K\subseteq[k]}\hat{f}(K)x^K, \text{where}~x^K=\prod_{i\in K}x_i
\end{equation}
\end{thm}
We assume that the scopes associated with the constraints are distinct when viewed as sets.
$ \mathbb{E}_x[\psi(x)]=\hat{\psi}(\emptyset) $ is the probability that a  constraint would be satisfied by a uniformly random assignment.

In some cases, once we fix the scope $ \mathcal{S}=\{x_1,\cdots,x_k\} $, there is additional freedom to choose what the constraint is, i.e. $ |\Psi|>1 $. For example, for Max-$ k $XOR, we can choose the constraint to be $ \psi(x_1,\cdots,x_k)=\frac{1}{2}\pm\frac{1}{2}\prod_{i=1}^{k}x_i $. In other cases, for example Max-Cut, we do not have such freedom. Once the scope is fixed, the associated constraint is also fixed. For general Max-CSP, we give a probability $ P(\psi_{i,l}) $ for each $ \psi_{i,l}\in\Psi_l $, with $ \sum_{i}P(\psi_{i,l})=1 $. Here $ l $ means it is the $ l $th constraint, and $ i $ indexes the possible predicates for the $ l $th constriant. All constraints are chosen independently according to some probability distributions defined above.
\begin{defn}\label{typicalitydef}
A constraint satisfaction problem is said to have \textit{typicality} if
\begin{equation}
\mathbb{E}_\psi[\psi(x)-\hat{\psi}(\emptyset)]\equiv\sum_iP(\psi_i)(\psi_i(x)-\hat{\psi}_i(\emptyset))=0, \quad\psi_i\in\Psi.
\end{equation}
\end{defn}
If we use the Fourier expansion of $ \psi(x) $, then Definition \ref{typicalitydef} is equivalent to 
\begin{equation}
\mathbb{E}_{\psi}[\hat{\psi}(K)]=0 ~~~\forall~K\neq\emptyset .
\end{equation}
In other words, the Fourier coefficients associated with a probability distribution of $ \psi $ have 0 mean.

Often, it is convenient to view $ \psi $ as a discrete random variable, with some underlying probability distribution. Therefore, we suppress the index $ i $ from now onwards, so as not to confuse with $ l $, which labels the constraints.

CSPs with typicality include $ k $XOR, in which every predicate is of the form $ \psi_{k\text{XOR}}(x_1,\dots,x_k)=\frac{1}{2}\pm\frac{1}{2}\prod_{i=1}^{k}x_i $ and $ k $SAT, in which every predicate is of the form $ \psi_{k\text{SAT}}(x_1,\dots,x_k)=1-\prod_{i=1}^{k}\frac{1\pm x_i}{2} $. An instance of $ k $XOR or $ k $SAT is constructed as follows: for each scope that we want to place a constraint, we choose the signs in each constraint independently at random. If the underlying constraint hypergraph has $ m $ constraints, we have $ 2^m $ possible $ k $XOR instances and $ 2^{km} $ possible $ k $SAT instances. For a general CSP, an instance is constructed by choosing the constraint on each scope according to some probability distribution, independently. This means that $ \hat{\psi_l}(S) $ from different constraints are independent random variables. This is similar to the semi-random model considered in \cite{Feige07}. Note that MaxCut does not have typicality, despite that a MaxCut instance is a 2XOR instance.

\begin{thm}\label{typicality}
Suppose a constraint satisfaction problem $ \Psi $ of bounded degree (Max-CSP$ (\Psi)_\text{B} $) has typicality. Given a set of scopes $\{\mathcal{S}_l\}$, construct a CSP instance $\{\psi_l,\mathcal{S}_l\}$ by choosing the constraints $\psi_l$ at random according to the probability distribution $P$. Then with probability at least $ 1 - O(D^3/m) $, there is a quantum algorithm which finds an assignment satisfying a $ \mu+\Omega(1/\sqrt{D}) $ fraction of the constraints.\footnote{Most statements on approximation algorithms are given in terms of approximation ratios: the number of contraints that the output assignment satisfies, divided by the number of constraints satisfied by the best assignment. Our statement, like that of previous work \cite{Farhi1412,Barak15}, is stronger; we show that for most instances, there is an assignment satisfying a $ \mu+\Omega(1/\sqrt{D}) $ fraction of the \emph{total} number of constraints.}
\end{thm}
We note that this theorem holds for all choices of the interaction graph $\{\mathcal{S}_l\}$; the probability in the theorem statement is taken only over the choice of the $\psi_l$'s. We typically consider the case when $m$ is much greater than $D$; for instance, in a $D$-regular graph, $ m = O (nD) $ where $n$ is the number of vertices. In that case $ D^3 / m \approx 0 $ for large instances, so the quantum algorithm works well for nearly all choices of constraints. 

We will prove this theorem in the next few sections. Moreover, we'll prove it for predicates $ \psi_l(x)=\sum_{K\subseteq K_l}\hat{\psi}_l(K)x^K $ in which the highest degree term has $ K=K_l $, i.e. the highest degree term covers the entire scope. The proof for the most general case is sketched in the appendix.

We also define what it means for an instance to have "no overlapping constraints", and show that a similar result applies to such instances of any CSP, using the same quantum algorithm.
\begin{thm}\label{triangle}
Suppose the instance of a constraint satisfaction problem of bounded degree $ D $ has "no overlapping constraints", then there is a quantum algorithm which finds an assignment satisfying $ \mu+\Omega(1/\sqrt{D}) $ fraction of the constraints.
\end{thm}

Throughout the paper, we use the Quantum Approximate Optimization Algorithm (QAOA), as introduced in \cite{Farhi1411}. The input is the instance of the CSP constraint, as specified by an objective function $ C(z)=\sum_l\psi_l(z) $. The instance has $ m $ constraints over $ n $ variables. A quantum state $ \ket{\gamma,\beta} $ is constructed as follows.

First, we take the initial state to be the uniform superposition over computational basis states $ \ket{z} $,
\begin{equation}
\ket{s}=\frac{1}{2^{n/2}}\sum_z\ket{z}.
\end{equation}

We then define
\begin{equation}
\ket{\gamma,\beta}=e^{-i\beta B}e^{-i\gamma C}\ket{s}
\end{equation}
where $ B=\sum_{i=1}^{n}X_i $ is the sum of all Pauli matrices $ X $ on single qubits, and $ C=C(Z) $ is the objective operator as defined by the objective function, with the variables $ z $ replaced by Pauli matrix $ Z $.

By varying the angles $ \beta $ and $ \gamma $, we vary the weights of different bit strings in $ \ket{\gamma,\beta} $. So if $ \beta=\gamma=0 $, then $ \bra{\gamma,\beta}C\ket{\gamma,\beta}=\mu m $, which is the expected number of constraints satisfied by a uniformly random assignment. The goal is to optimize
\begin{equation}\label{qexp}
\bra{\gamma,\beta}C\ket{\gamma,\beta},
\end{equation}
by picking $ \beta $ and $ \gamma $ wisely. Once we've found good choices for $ \beta $ and $ \gamma $, we can measure $ \ket{\gamma,\beta} $ in the computational basis, to find actual assignments $ \ket{z} $ which optimize $ C(z) $.

We first show how QAOA can be implemented on Max-2XOR to obtain an advantage of $ \Omega(1/\sqrt{D}) $ for the typical cases, assuming we construct an instance by picking the sign in each constraint $ \psi(x)=\frac{1}{2}\pm\frac{1}{2}x_ix_j $ to be random.

\section{QAOA for MAX2XOR}\label{sec:MAX2XOR}
Firstly, we look at Max-2XOR where each variable appears in at most $ D+1 $ constraints. The $ +1 $ is for later convenience.

For Max-2XOR, the constraints are of the form $ \psi(z)=\frac{1}{2}+\frac{1}{2}d_{ij}z_iz_j $, where $ d_{ij}=\pm1 $ with probability 1/2 for each sign. WLOG, we look at an instance involving bits 1 and 2. We separate out the term $\frac{1}{2}d_{12}Z_1Z_2$ from the rest of the clauses:
\begin{equation}
C = \bar{C} + \frac{1}{2}d_{12}Z_1Z_2.
\end{equation}
The contribution of this term to the overall expectation \ref{qexp} is
\begin{align}
&\frac{1}{2}\bra{s}e^{i\gamma C}e^{i\beta B}d_{12}Z_1Z_2e^{-i\beta B}e^{-i\gamma C}\ket{s}\label{MAX2XOR}\\
=&\frac{1}{2}d_{12}\bra{s}e^{i\gamma\bar{C}}e^{i\gamma d_{12}Z_1Z_2/2}(\cos(2\beta)Z_1+\sin(2\beta)Y_1)(\cos(2\beta)Z_2+\sin(2\beta)Y_2)e^{-i\gamma d_{12}Z_1Z_2/2}e^{-i\gamma\bar{C}}\ket{s}\nonumber
\end{align}
where
\begin{align}
&d_{12}e^{i\gamma d_{12}Z_1Z_2/2}(\cos(2\beta)Z_1+\sin(2\beta)Y_1)(\cos(2\beta)Z_2+\sin(2\beta)Y_2)e^{-i\gamma d_{12}Z_1Z_2/2}\\
=&d_{12}(\cos^2(2\beta)Z_1Z_2+\sin^2(2\beta)Y_1Y_2+\sin(2\beta)\cos(2\beta)\cos\gamma(Z_1Y_2+Y_1Z_2))\nonumber\\
+&\sin(2\beta)\cos(2\beta)\sin\gamma(X_1+X_2)\nonumber
\end{align}

The relevant terms in $ \bar{C} $ are $ Z_1C_1+Z_2C_2 $, where
\begin{align}
&C_1=\frac{1}{2}\sum_{b\neq 1,2}d_{1b}Z_b\\
&C_2=\frac{1}{2}\sum_{b\neq 1,2}d_{2b}Z_b
\end{align}

So (\ref{MAX2XOR}) becomes
\begin{equation}
\frac{1}{2}\bra{s}\left(d_{12}\sin^2(2\beta)\sin(2\gamma C_1)\sin(2\gamma C_2)+\sin(2\beta)\cos(2\beta)\sin\gamma(\cos(2\gamma C_1)+\cos(2\gamma C_2))\right)\ket{s}.\label{result}
\end{equation}


Assuming that $ C_1 $ has $ D_1 $ terms and $ C_2 $ has $ D_2 $ terms,
\begin{align}
&\bra{s}\cos(2\gamma C_1)\ket{s}=(\cos\gamma)^{D_1}\\
&\bra{s}\cos(2\gamma C_2)\ket{s}=(\cos\gamma)^{D_2}
\end{align}

If we average over the clauses (i.e. the $ d_{ab} $'s), the first term in (\ref{result}) drops out. By picking $ \beta=\frac{\pi}{8} $, $ \gamma=\frac{g}{\sqrt{D}} $, and remembering $ D_1,D_2\leq D $, we then obtain that the average of (\ref{result})
is, in the large $ D $ limit, at least
\begin{equation}
\frac{1}{2}\frac{g}{\sqrt{D}}\exp\left(-\frac{g^2}{2}\right)
\end{equation}
which can be maximized to
\begin{equation}
\frac{1}{2\sqrt{e}\sqrt{D}}
\end{equation}
if we take $ g=1 $.

Summing over all the constraints, this means when averaged over the possible choices of the clauses, our quantum algorithm satisfies an expected fraction of
\begin{equation}
\left(\frac{1}{2}+\frac{1}{2\sqrt{e}\sqrt{D}}\right)
\end{equation}
of the constraints, or equivalently a total number of
\begin{equation}
\left(\frac{1}{2}+\frac{1}{2\sqrt{e}\sqrt{D}}\right)m
\end{equation}
constraints. As we will show in Section \ref{subsec:variance}, the variance of this quantity with respect to the choice of the clauses is $ O(mD^2) $, and therefore for most instances of the CSP this ratio is achievable.

Note that this does not imply we can optimize every Max-2XOR instance to $ \Omega(1/\sqrt{D}) $. For example, consider Max-Cut on copies of $ (D+1) $-cliques. The advantage can only be $ \Omega(1/D) $. This is contrary to the case of Max-3XOR: Barak \textit{et al}. showed an advantage of $ \Omega(1/\sqrt{D}) $ can be obtained for every Max-3XOR instance using a classical algorithm\cite{Barak15}, whereas Farhi \textit{et al}. showed an advantage of $ \Omega(1/\sqrt{D}) $ can be obtained for typical Max-3XOR instances using QAOA\cite{Farhi1412}. In a sense, Max-Cut is the \textit{worst} case instance for Max-2XOR.

Barak \textit{et al}. pointed out that for Max-2XOR on any $D$-regular graph on $ n $ vertices, if we construct a Max-2XOR instance by choosing the constraint on each edge to be random, with high probability, all assignments $x$ will have $ \left|\text{val}(x)-\frac{1}{2}\right|\leq O(1/\sqrt{D}) $\cite{Barak15}. On the other hand, our result (Theorem \ref{typicality}) holds regardless of the underlying graph. Hence, for almost all instances constructed from an underlying constraint hypergraph, we can find an assignment achieving an advantage of $ \Omega(1/\sqrt{D}) $. This is optimal by the above argument.
\section{QAOA for CSPs with typicality}\label{sec:typicality}
In this section, we first establish the result for \textit{generalized} Max-XOR, then extend our result to general CSPs. By \textit{generalized} Max-XOR, we mean $ \psi_l(x)=\mu+\hat{\psi}_l(K_l)x^{K_l} $, where $ K_l $ is a subset of $ [n] $ with $ |K_l|=k_l\leq k $. $ k_l $ may be different for different $ l $. For notational simplicity, sometimes we write $ x^{K_l}=x_{i_1}\cdots x_{i_{k_l}} $ and $ \hat{\psi}_l(K_l)=\hat{\psi}_l $ for $ K_l=\{i_1,\dots,i_{k_l}\} $. Here we've implicitly relaxed the image of $ \psi(x) $ to $ \mathbb{R} $. To recover Max-$ k $XOR, one just chooses $ \hat{\psi}_l(K_l) $ to be $ \pm1/2 $ at random, and $ k_l=k $ for all $ l $. Here, we also relax $ \hat{\psi}_l(K_l) $ to be some general random variable with 0 mean.
\subsection{Max-XOR}\label{subsec:XOR}
\begin{lem}
For generalized Max-XOR, there is a quantum algorithm which finds an assignment satisfying $ \mu+\Omega(1/\sqrt{D}) $ of the constraints for typical instances.
\end{lem}
The objective operator is 
\begin{equation}
\hat{\psi}_lZ^{K_l}
\end{equation}
where we've dropped $ \mu $. Given parameters $ \beta $ and $ \gamma $, we define the state
\begin{equation}
\ket{\gamma,\beta}=e^{-i\beta B}e^{-i\gamma C}\ket{s}
\end{equation}
where $ B=X_1+\cdots+X_n $ and $ C=\sum_l\hat{\psi}_lZ^{K_l} $. We wish to evaluate
\begin{equation}
\mathbb{E}_\psi\left[\bra{\gamma,\beta}C\ket{\gamma,\beta}\right]\label{MaxkXORoddexp}
\end{equation}
for some fixed values of $ \gamma $ and $ \beta $. Here $ \mathbb{E}_\psi\equiv\mathbb{E}_{\hat{\psi}_1,\dots,\hat{\psi}_m} $, where $ \hat{\psi}_i $'s are independent random variables.

The $ u $-th term in the quantum expectation (\ref{MaxkXORoddexp}) is
\begin{align}\label{koneterm}
&\mathbb{E}_\psi\left[\bra{s}e^{i\gamma C}e^{i\beta B}\hat{\psi}_uZ_{i_1}\cdots Z_{i_{k_u}}e^{-i\beta B}e^{-i\gamma C}\ket{s}\right]\\
=&\mathbb{E}_\psi\left[\bra{s}e^{i\gamma \bar{C}}\exp(i\gamma \hat{\psi}_uZ_{i_1}\cdots Z_{i_{k_u}})\hat{\psi}_u\prod_{i\in K_u}(\cos(2\beta)Z_i+\sin(2\beta)Y_i)\exp(-i\gamma \hat{\psi}_uZ_{i_1}\cdots Z_{i_{k_u}})e^{-i\gamma \bar{C}}\ket{s}\right]\nonumber
\end{align}
where $ \bar{C}=\sum_{l\neq u}\hat{\psi}_lZ^{K_l} $, and
\begin{align}\label{product}
	&\hat{\psi}_u\prod_{i\in K_u}(\cos(2\beta)Z_i+\sin(2\beta)Y_i)\\
	=&\sum_{\substack{p+q=k_u\\s_1,\dots,s_p,t_1,\dots,t_q\in K_u}}\cos(2\beta)^p\sin(2\beta)^q\hat{\psi}_uZ_{s_1}\cdots Z_{s_p}Y_{t_1}\cdots Y_{t_q}.\nonumber
\end{align}
If $ q $ is even, then
\begin{equation}\label{even}
	\exp(i\gamma \hat{\psi}_uZ_{i_1}\cdots Z_{i_{k_u}})\hat{\psi}_uZ_{s_1}\cdots Z_{s_p}Y_{t_1}\cdots Y_{t_q}\exp(-i\gamma \hat{\psi}_uZ_{i_1}\cdots Z_{i_{k_u}})=\hat{\psi}_uZ_{s_1}\cdots Z_{s_p}Y_{t_1}\cdots Y_{t_q}
\end{equation}
and these terms do not contribute to the expectation in (\ref{koneterm}).\\
If $ q $ is odd, then
\begin{align}\label{odd}
	&\exp(i\gamma\hat{\psi}_uZ_{i_1}\cdots Z_{i_{k_u}})\psi_uZ_{s_1}\cdots Z_{s_p}Y_{t_1}\cdots Y_{t_q}\exp(-i\gamma \hat{\psi}_uZ_{i_1}\cdots Z_{i_{k_u}})\\
	=&\hat{\psi}_u\cos(2\gamma\hat{\psi}_u) Z_{s_1}\cdots Z_{s_p}Y_{t_1}\cdots Y_{t_q}-i^{q+1}\hat{\psi}_u\sin(2\gamma\hat{\psi}_u) X_{t_1}\cdots X_{t_q}\nonumber
\end{align}
and only the second term contributes to (\ref{koneterm}) in leading order.\\
Take $ \gamma=g/\sqrt{D} $. The expectation of (\ref{koneterm}) over $ \hat{\psi}_u $ is
\begin{align}
&\frac{2g\mathrm{Var}[\hat{\psi}_u]}{\sqrt{D}}\sum_{\substack{q ~\text{odd}\\\{t_1,\dots t_q\}\subset K_u}}-i^{q+1}\cos(2\beta)^{k_u-q}\sin(2\beta)^q\bra{s}e^{i\gamma\bar{C}}X_{t_1}\cdots X_{t_q}e^{-i\gamma\bar{C}}\ket{s}+O\left(\frac{1}{D}\right)\nonumber\\
\geq&\frac{2g\mathrm{Var}[\hat{\psi}_u]}{\sqrt{D}}\bigg(\cos(2\beta)^{k_u-1}\sin(2\beta)\bra{s}e^{i\gamma\bar{C}}(X_{i_1}+\cdots +X_{i_{k_u}})e^{-i\gamma\bar{C}}\ket{s}\nonumber\\
&~~~~~~~~~~~~~-\sum_{q=3, ~q~\text{odd}}^{2\lfloor(k_u-1)/2\rfloor+1}|\cos(2\beta)^{k_u-q}\sin(2\beta)^q|\binom{k_u}{q}\bigg)+O\left(\frac{1}{D}\right)\label{kave}
\end{align}
where we've used
\begin{equation}
	\left|\bra{s}e^{i\gamma\bar{C}}X_{t_1}\cdots X_{t_q}e^{-i\gamma\bar{C}}\ket{s}\right|\leq 1
\end{equation}
and
\begin{align}
&\mathbb{E}_\psi[\cos(2\gamma \hat{\psi}_u)]=1-\frac{2g^2\mathrm{Var[\hat{\psi}_u]}}{D}+O\left(\frac{1}{D^2}\right)\\
&\mathbb{E}_\psi[\sin(2\gamma \hat{\psi}_u)]=O\left(\frac{1}{D^{3/2}}\right)\\
&\mathbb{E}_\psi [\hat{\psi}_u\cos(2\gamma \hat{\psi}_u)]=O\left(\frac{1}{D}\right)\label{psicos}\\
&\mathbb{E}_\psi [\hat{\psi}_u\sin(2\gamma \hat{\psi}_u)]=\frac{2g\mathrm{Var}[\hat{\psi}_u]}{\sqrt{D}}+O\left(\frac{1}{D^{3/2}}\right)\label{psisin}.
\end{align}
Define
\begin{equation}
G_u(j)=\{l: \mathcal{S}_l \ni x_j \;\text{with} \;l\neq u\}
\end{equation}
These are the constraints which overlap $ \mathcal{S}_u $ at $ x_j $. We've assumed bounded occurrence, so each bit is in no more than $ D+1 $ clauses. 
Since $ \hat{\psi}_l $ and $ \hat{\psi}_{l^\prime} $ are independent random variables,
\begin{align}\label{kexp1}
&\mathbb{E}_\psi[\bra{s}e^{i\gamma\bar{C}}X_je^{-i\gamma\bar{C}}\ket{s}]\\
=&\bra{s}\mathbb{E}_\psi\left[\prod_{l\in G_u(j)}[\cos(2\gamma \hat{\psi}_l)+i\sin(2\gamma \hat{\psi}_l)Z^{K_l}]X_j\right]\ket{s}\nonumber\\
=&\bra{s}\prod_{l\in G_u(j)}\left(1-\frac{2g^2\mathrm{Var}[\hat{\psi}_l]}{D}+O\left(\frac{1}{D^{3/2}}\right)Z^{K_l}\right)X_j\ket{s}+O\left(\frac{1}{D^2}\right)\nonumber\\
=&\prod_{l\in G_u(j)}\left(1-\frac{2g^2\mathrm{Var}[\hat{\psi}_l]}{D}\right)+O\left(\frac{1}{D}\right)\nonumber
\end{align}
where we've used the fact that $ \bra{s}Z^{K_l}X_j\ket{s}=0 $.

Since
\begin{equation}
\prod_{l\in G_u(j)}\left(1-\frac{2g^2\mathrm{Var}[\hat{\psi}_l]}{D}\right)
\geq\exp(-2g^2\text{max}\mathrm{Var}[\hat{\psi}_l])+O\left(\frac{1}{D}\right),
\end{equation}
the expectation of (\ref{kave}) over $ \psi $ is lower bounded as
\begin{align}\label{kave1}
	&\mathbb{E}_\psi\left[\bra{s}e^{i\gamma C}e^{i\beta B}\hat{\psi}_uZ_{i_1}\cdots Z_{i_{k_u}}e^{-i\gamma B}e^{-i\gamma C}\ket{s}\right]\\
	\geq&\frac{2g\mathrm{Var}[\hat{\psi}_u]}{\sqrt{D}}\left(\cos(2\beta)^{k_u-1}\sin(2\beta)k_u\exp(-2g^2\text{max}\mathrm{Var}[\hat{\psi}_l])-\sum_{q=3,~q~\text{odd}}^{2\lfloor(k_u-1)/2\rfloor+1}|\cos(2\beta)^{k_u-q}\sin(2\beta)^q|\binom{k_u}{q}\right)\nonumber\\
	&+O\left(\frac{1}{D}\right).\nonumber
\end{align}

In the second last line, if one expands in $ \beta $, the term linear in $ \beta $ has positive coefficient. So one can always optimize $ \beta $ (which would depend on $ k_u $ but not $ D $) to make the coefficient of $ \frac{1}{\sqrt{D}} $ positive, so that the average of (\ref{koneterm}) is bounded below by
\begin{equation}
	\frac{c}{\sqrt{D}}
\end{equation}
for some constant.
\subsection{General CSPs with typicality}\label{subsec:typicality}
Now we show Theorem \ref{typicality} can be achieved with QAOA.
Consider a general predicate of the form
\begin{equation}
\psi_l(Z)=\sum_{K\subseteq K_l}\hat{\psi}_l(K)Z^K=\sum_{K\subset K_l}\hat{\psi}_l(K)Z^K+\hat{\psi}_lZ^{K_l}
\end{equation}
in which the second summation is over all proper subsets of $ K_l $. We assume that $ \hat{\psi}_l$ is not identically zero; for the general case, please refer to the appendix.

We'll evolve with the truncated Hamiltonian
\begin{equation}
\ket{\gamma,\beta}=e^{-i\beta B}e^{-i\gamma C}\ket{s}
\end{equation}
with $ B $ the same as before, but now $ C=\sum_l\hat{\psi}_lZ^{K_l} $. $ C $ is the sum of the highest degree terms from $ \psi(Z) $. Note that in the original formulation of QAOA, one has to evolve the full Hamiltonian, i.e. take $ C=\sum_l\psi_l(Z) $ instead.

We wish to evaluate
\begin{equation}\label{oddgeneralexp}
\mathbb{E}_\psi\left[\bra{\gamma,\beta}\sum_l\psi_l(Z)\ket{\gamma,\beta}\right].
\end{equation}

Setting $ \gamma=1/\sqrt{D} $, the contribution from the highest degree terms of $ \sum_l\psi_l(Z) $ to (\ref{oddgeneralexp}) is
\begin{equation}
\mathbb{E}_\psi\left[\bra{\gamma,\beta}\sum_l\hat{\psi}_lZ^{K_l}\ket{\gamma,\beta}\right].
\end{equation}
This is identical to that of \textit{generalized} Max-XOR. Now we show the contribution of the remaining terms is at most of order $ O(1/D) $.

WLOG, look at a lower degree term of the form $ \hat{\psi}_u(K)Z^K $ for $ K\subsetneq K_u $. It contributes as
\begin{equation}\label{oddgeneralsub}
\mathbb{E}_\psi\left[\bra{s}e^{i\gamma C}e^{i\beta B}\hat{\psi}_u(K)Z^Ke^{-i\beta B}e^{-i\gamma C}\ket{s}\right].
\end{equation}

As in (\ref{product}), it breaks up into terms of the form
\begin{equation}\label{generalcsponeterm}
\cos(2\beta)^p\sin(2\beta)^q\mathbb{E}_\psi\left[\hat{\psi}_u(K)\bra{s}e^{i\gamma \bar{C}}e^{i\gamma\hat{\psi}_uZ^{K_u}}Z_{s_1}\cdots Z_{s_p}Y_{t_1}\cdots Y_{t_q}e^{-i\gamma\hat{\psi}_uZ^{K_u}}e^{-i\gamma C}\ket{s}\right]
\end{equation}
where $ s_1,\dots,s_p,t_1,\dots,t_q\in K $, $ p+q=|K| $, and $ \bar{C}=\sum_{l\neq u}\hat{\psi}_lZ^{K_l} $. Since $ K\subsetneq K_u $, $ |K|<k_u $.

Note that $ \hat{\psi}_u $ and $ \hat{\psi}_u(K) $ need not be independent, so $ \mathbb{E}_\psi[\hat{\psi}_u(K)]=0 $ does not imply these terms do not contribute.

If $ q $ is even, then $ Z^{K_u} $ and $ Z_{s_1}\cdots Z_{s_p}Y_{t_1}\cdots Y_{t_q} $ commute and (\ref{oddgeneralsub}) becomes
\begin{equation}
\cos(2\beta)^p\sin(2\beta)^q\mathbb{E}_\psi\left[\hat{\psi}_u(K)\bra{s}e^{i\gamma\bar{C}} Z_{s_1}\cdots Z_{s_p}Y_{t_1}\cdots Y_{t_q}e^{-i\gamma\bar{C}}\ket{s}\right].
\end{equation}
This vanishes, since $ \mathbb{E}_\psi[\hat{\psi}_u(K)]=0 $.

If $ q $ is odd, then (\ref{generalcsponeterm}) becomes
\begin{align}\label{qodd}
&\cos(2\beta)^p\sin(2\beta)^q\mathbb{E}_\psi\bigl[\hat{\psi}_u(K)\bra{s}e^{i\gamma\bar{C}}(\cos(2\gamma\hat{\psi}_u)Z_{s_1}\cdots Z_{s_p}Y_{t_1}\cdots Y_{t_q}\\
&~~~~~~~~~~~~~~~~~~~~~~~~~~~~~~~~~~~~-i^{q+1}\sin(2\gamma\hat{\psi}_u)X_{t_1}\cdots X_{t_q}Z^{K_u\setminus K})e^{-i\gamma\bar{C}}\ket{s}\bigr].\nonumber
\end{align}

The first term ($ \cos(2\gamma\hat{\psi}_u) $ term) of (\ref{qodd}) is at most $ O(1/D) $, since $ \mathbb{E}_\psi[\hat{\psi}_u(K)\cos(2\gamma\hat{\psi}_u)]=O(1/D) $. The second term ($ \sin(2\gamma\hat{\psi}_u) $ term) gives a contribution of
\begin{equation}
\cos(2\beta)^p\sin(2\beta)^q\frac{-i^{q+1}2g\mathbb{E}_\psi[\hat{\psi}_u(K)\hat{\psi}_u]}{\sqrt{D}}\mathbb{E}_\psi\left[\bra{s}e^{i\gamma\bar{C}}X_{t_1}\cdots X_{t_q}Z^{K_u\setminus K}e^{-i\gamma\bar{C}}\ket{s}\right]+O\left(\frac{1}{D^{3/2}}\right),
\end{equation}
since
\begin{equation}
\mathbb{E}_\psi[\hat{\psi}_u(K)\sin(2\gamma\hat{\psi}_u)]=\frac{2g\mathbb{E}_\psi[\hat{\psi}_u(K)\hat{\psi}_u]}{\sqrt{D}}+O\left(\frac{1}{D^{3/2}}\right).
\end{equation}
Only the terms in $ \bar{C} $ that overlap with $ X_{t_1}\cdots X_{t_q} $ at odd number of bits contribute. So WLOG we assume $ \bar{C} $ only contains those terms.
\begin{align}
&\mathbb{E}_\psi\left[\bra{s}e^{i\gamma\bar{C}}X_{t_1}\cdots X_{t_q}Z^{K_u\setminus K}e^{-i\gamma\bar{C}}\ket{s}\right]\\
=&\mathbb{E}_\psi\left[\bra{s}\prod_l(\cos(2\gamma\hat{\psi}_l)+i\sin(2\gamma\hat{\psi}_l)Z^{K_l})X_{t_1}\cdots X_{t_q} Z^{K_u\setminus K}\ket{s}\right]\nonumber\\
=&\bra{s}\prod_l\left(1-\frac{2g^2\mathrm{Var}[\hat{\psi}_l]}{D}+iO\left(\frac{1}{D^{3/2}}\right)Z^{K_l}\right)X_{t_1}\cdots X_{t_q} Z^{K_u\setminus K}\ket{s}+O\left(\frac{1}{D^2}\right)\nonumber\\
=&O\left(\frac{1}{D}\right)\nonumber
\end{align}
since $ \bra{s}X_{t_1}\cdots X_{t_q} Z^{K_u\setminus K}\ket{s}=0 $ and $ \bra{s}Z^{K_l}X_{t_1}\cdots X_{t_q} Z^{K_u\setminus K}\ket{s}=0 $.

Therefore the second term in (\ref{qodd}) gives a contribution of $ O(1/D^{3/2}) $, and the total contribution of a lower degree term as in (\ref{generalcsponeterm}) is $ O(1/D) $. There are at most $ 2^k-2 $ lower degree terms in a predicate, which is constant if we treat $k$ as constant. So the contribution of all the lower degree terms is still $ O(1/D) $.

\subsection{Variance}\label{subsec:variance}
We've just shown that
\begin{equation}
	\mathbb{E}_\psi\left[\bra{\gamma,\beta}\psi_l(Z)-\mu\ket{\gamma,\beta}\right]=\Omega\left(\frac{1}{\sqrt{D}}\right),
\end{equation}
i.e.
\begin{equation}
	\mathbb{E}_\psi\left[\bra{\gamma,\beta}\sum_{l=1}^{m}\psi_l(Z)\ket{\gamma,\beta}\right]=\left(\mu+\Omega\left(\frac{1}{\sqrt{D}}\right)\right)m.
\end{equation}

Now we show that the variance of the quantity inside the expectation above, with respect to $ \psi $, is small. The variance is the sum of 
\begin{equation}\label{variance1}
\mathbb{E}_\psi\left[\bra{\gamma,\beta}\psi_l(Z)\ket{\gamma,\beta}\bra{\gamma,\beta}\psi_{l^\prime}(Z)\ket{\gamma,\beta}\right]-\mathbb{E}_\psi\left[\bra{\gamma,\beta}\psi_l(Z)\ket{\gamma,\beta}\right]\mathbb{E}_\psi\left[\bra{\gamma,\beta}\psi_{l^\prime}(Z)\ket{\gamma,\beta}\right]
\end{equation}
over $ l,l^\prime $.

Fix $ l $. Since we've assumed bounded degree, $ \bra{\gamma,\beta}\psi_l(Z)\ket{\gamma,\beta} $ involves at most $ k(k-1)D+k $ bits. Each of these bits is in at most $ D+1 $ clauses. So $ \bra{\gamma,\beta}\psi_l(Z)\ket{\gamma,\beta} $ is "linked" to at most $ (k(k-1)D+k)(D+1) $ clauses. If $ \psi_{l^\prime} $ is not one of these clauses, then (\ref{variance1}) is 0. Otherwise, since $ \psi_l(x)=0 $ or 1, (\ref{variance1}) is bounded by 1. So
\begin{equation}
\mathrm{Var}_{\psi}\left[\bra{\gamma,\beta}\sum_{l=1}^{m}\psi_l(Z)\ket{\gamma,\beta}\right]\leq m(k(k-1)D+k)(D+1).
\end{equation}
i.e, the standard deviation is $ O(\sqrt{m}D) $.

By Chebyshev's inequality, this implies that with probability at least $ 1 - O (D^3/m) $ over the choice of constraints, our choice of angles $ \beta $ and $ \gamma $ gives
\begin{align}
\bra{\gamma,\beta}\sum_{l=1}^{m}\psi_l(Z)\ket{\gamma,\beta}&=\left(\mu+\Omega\left(\frac{1}{\sqrt{D}}\right)\right)m\pm \sqrt{\frac{m}{D^3}}O(\sqrt{m}D) \\
&= \left(\mu+\Omega\left(\frac{1}{\sqrt{D}}\right)\right)m.
\end{align}


\begin{cor}
	For Max-$ k $XOR and Max-$ k $SAT, if we choose the signs in each constraint to be random and independent, then there is a quantum algorithm which finds an assignment satisfying $ \mu+\Omega(1/\sqrt{D}) $ of the constraints for typical instances.
\end{cor}

\subsection{$ k $-local Hamiltonian of bounded degree}
More recently, Harrow and Montanaro applied the classical algorithm to find the extremal eigenvalues of local Hamiltonians. They considered the "quantum analogue" of CSPs with bounded degree, and obtained the following result.
\begin{thm}[{\cite[Theorem~1]{Harrow1507}}]
	Let $ H $ be a traceless $k$-local Hamiltonian on $ n $ qubits such that $ k=O(1) $. Assume that $ H $ can be expressed as a weighted sum of $ m $ distinct Pauli terms such that each term is weight $\Theta(1)$, and each qubit participates in at most $ l $ terms. Then $ ||H||=\Omega(m/\sqrt{l}) $ and $ \lambda_{\mathrm{min}}(H)\leq -\Omega(m/l) $. In each case the bound is achieved by a product state which can be found efficiently using a classical algorithm.
\end{thm}

Using our algorithms (Theorem \ref{typicality}), it is immediate that the quantum-classical correspondence of \cite{Harrow1507} can be strengthened to give the typical ground state energies of such $ k $-local Hamiltonians.
\begin{cor}
	Let $ H $ be a traceless $k$-local Hamiltonian on $ n $ qubits such that $ k=O(1) $. Assume that $ H $ can be expressed as a weighted sum of $ m $ distinct Pauli terms such that each term is weight $\Theta(1)$, and each qubit participates in at most $ l $ terms. The weight is chosen independently at random, with zero mean. Then with high probability, $ \lambda_{\mathrm{min}}(H)\leq -\Omega(m/\sqrt{l}) $. This can be achieved by a product state which can be found efficiently using a quantum algorithm.
\end{cor}
\section{Instances with "No Overlapping Constraints"}\label{sec:nooverlapping}
In this section, we present the proof of Theorem \ref{triangle}. We first define what we mean by "no overlapping constraints", then give the construction of our quantum algorithm.
\begin{defn}[{\cite[Definition~2.1]{Barak15}}]\label{nooverlapping}
An instance has \textit{no overlapping constraints} if the scopes of any two distinct constraints intersect on at most one variable.
\end{defn}

We claim that the quantum algorithm that we used above can give a similar result for instances with "no overlapping constraints".

Consider such an instance in which the highest degree term covers the entire scope. The general case will be similar.
As before, we'll evolve with the truncated Hamiltonian.
The contribution of a highest degree term is
\begin{align}\label{trianglefreeleading}
	&\bra{\gamma,\beta}\hat{\psi}_uZ^{K_u}\ket{\gamma,\beta}\\
	=&\bra{s}e^{i\gamma \bar{C}}\exp(i\gamma \hat{\psi}_uZ_{i_1}\cdots Z_{i_{k_u}})\hat{\psi}_u\prod_{i={i_1}}^{i_{k_u}}(\cos(2\beta)Z_i+\sin(2\beta)Y_i)\exp(-i\gamma \hat{\psi}_uZ_{i_1}\cdots Z_{i_{k_u}})e^{-i\gamma \bar{C}}\ket{s}.\nonumber
\end{align}

By "no overlapping constraints", the relevant terms in $ \bar{C} $ are of the form $ Z_jC_j $ only, where $ C_j=\partial_j(\sum_{l\neq u}\hat{\psi}_lZ^{K_l}) $ does not share the same variable as $ Z^{K_u} $. As $ \bra{+}Z\ket{+}=0 $, (\ref{trianglefreeleading}) becomes
\begin{align}\label{kxortriangle}
	&\hat{\psi}_u\sin(2\gamma\hat{\psi}_u)\sum_{\substack{q ~\text{odd}\\\{t_1,\dots t_q\}\subset K_u}}-i^{q+1}\cos(2\beta)^{k_u-q}\sin(2\beta)^q\bra{s}\cos(2\gamma C_{t_1})\cdots\cos(2\gamma C_{t_q})\ket{s}\\
	\geq&\hat{\psi}_u\sin(2\gamma\hat{\psi}_u)\biggl(\cos(2\beta)^{k_u-1}\sin(2\beta)\bra{s}(\cos(2\gamma C_{i_1})+\cdots+\cos(2\gamma C_{i_{k_u}}))\ket{s}\nonumber\\
	&~~~~~~~~~~~~~~~-\sum_{q=3, ~q~\text{odd}}^{2\lfloor(k_u-1)/2\rfloor+1}|\cos(2\beta)^{k_u-q}\sin(2\beta)^q|\binom{k_u}{q}\biggl).\nonumber
\end{align}

Again by "no overlapping constraints", the terms in $ C_j $ do not share the same variables. So
\begin{equation}\label{oneC}
	\bra{s}\cos(2\gamma C_j\ket{s}=\prod_{l\in G_u(j)}\cos(2\gamma\hat{\psi}_l).
\end{equation}

Take $ \gamma=g/\sqrt{D} $, (\ref{oneC}) can be lower bounded by
\begin{equation}
	\exp(-2g^2\text{max}\hat{\psi}_l^2)+O\left(\frac{1}{D}\right)
\end{equation}
and it can be shown that (\ref{kxortriangle}) is $ \frac{c}{\sqrt{D}} $ with some choice of $ g $ and $ \beta $.

It can be shown that lower degree terms do not contribute.

\section{Classical Algorithm}\label{sec:classical}
In this section, we show the classical algorithm by Barak \textit{et al.}\cite{Barak15} for triangle-free instances gives an assignment that, when averaged over all choices of constraints, satisfies a $\mu + \Omega(1/\sqrt{D})$ fraction of the constraints, for the case in which the highest degree term covers the entire scope. It remains unclear whether this algorithm can be applied to the most general case. Moreover, we do not know if this algorithm satisfies a $\mu + \Omega(1/\sqrt{D})$ fraction of the constraints on most instances; one way to show this might be to prove that the variance of the number of constraints satisfied is small. 

To start off, we first need the following results from Fourier analysis.

\begin{thm}[{\cite[Theorem~9.24]{Odonnell14}}]\label{inequality}
	Let $ f:\{-1,1\}^n\to \mathrm{R} $ be a non-constant function of degree at most $ k $. Then
	\begin{equation}
		\mathbf{P}\left[f(x)>\mathbb{E}_x[f]\right]\geq\frac{1}{4}e^{-2k}.
	\end{equation}
\end{thm}

Applying it to $ f^2(x) $, which has degree at most $ 2k $, we have
\begin{equation}
	\mathbf{P}[|f(x)|>||f||_2]\geq\frac{1}{4}e^{-4k}
\end{equation}
and thus
\begin{equation}
	\mathbb{E}_x[|f(x)|]\geq\frac{1}{4}e^{-4k}||f||_2.
\end{equation}

\begin{lem}[{\cite[Lemma~2.2]{Barak15}}]\label{var}
	For any predicate $ \psi: \{-1,1\}^k\to \{0,1\} $, $ k\geq2 $, we have $ \mathrm{Var}[\partial_i\psi(x)]\geq2^{-k-2} $.
\end{lem}

\begin{proof}
The variables are divided into two partitions, with $ F $ the "fixed" part, and $ G $ the "greedy" part. The partition is fixed for now but will be chosen at random later. The assignments of variables in $ F $ are chosen at random, whereas those in $ G $ will be chosen according to those in $ F $.
	
The constraints are categorized into active ones and inactive ones. A constraint $ (\psi_l,\mathcal{S}_l) $ is \textit{active} if $ \mathcal{S}_l $ has exactly one variable in $ G $. Since each constraint has only one variable in $ G $, we can further partition the active constraints according to $ G $. If $ x_j\in G $, define
\begin{equation}
N_j=\{l: \mathcal{S}_l~\text{is active and}~x_j\in \mathcal{S}_l\}.
\end{equation}
and
\begin{equation}
A_j=\bigcup_{l\in N_j}\{\mathcal{S}_l\setminus\{x_j\}\}.
\end{equation}
	
We'll choose $ x_j $ in such a way that it only depends on $ \psi_l $ with $ l\in N_j $ and $ A_j $.
	
By construction, it is immediate that the inactive constraints contribute nothing on average, i.e.
\begin{equation}
\mathbb{E}_{\psi_l}[\mathbb{E}_x[\psi_l(x)]]=0~~\text{if}~\psi_l~\text{is inactive.}
\end{equation}
	
For each $ l\in N_j $, write $ \psi_l(x)=x_jQ_l(x)+R_l(x)+\hat{\psi}_l(\emptyset) $, where $ Q_l(x)=\partial_j\psi_l(x) $ and $ R_l(x) $ are now functions of $ \mathcal{S}_l\setminus\{x_j\} $.
	
Since $ R_l(x) $ only depends on $ \mathcal{S}_l\setminus\{x_j\}\subset F $, and the variables in $ F $ are chosen at random, we have
\begin{equation}
\mathbb{E}_x[R_l(x)]=0
\end{equation}
and thus
\begin{equation}
\mathbb{E}_{\psi_l}[\mathbb{E}_x[R_l(x)]]=0.
\end{equation}
	
Now we describe how to choose $ x_j $. Define $ \tilde{Q}_j=\sum_{l\in N_j}Q_l $ and $ \theta_j=\mathbb{E}_x[\tilde{Q}_j] $. So $ \theta_j $ is the mean of $ \tilde{Q}_j $ over the random inputs $ x $. With these, we can define\footnote{See \cite{Barak15} for some technical remarks on this.}
\begin{equation}
x_j=\text{sgn}\left(\tilde{Q}_j-\theta_j\right).
\end{equation}
So the net contribution from all the constraints in $ N_j $, when averaged over $ \psi_l $'s, is
\begin{equation}
\mathbb{E}_\psi\biggl[\mathbb{E}_x\biggl[\sum_{l\in N_j}\psi_l(x)-\hat{\psi}_l(\emptyset)\biggr]\biggr]=\mathbb{E}_\psi[\mathbb{E}_x[x_j\tilde{Q}_j]],
\end{equation}
where
\begin{equation}
\mathbb{E}_x[x_j\tilde{Q}_j]=\mathbb{E}_x[\text{sgn}(\tilde{Q}_j-\theta_j)(\tilde{Q}_j-\theta_j+\theta_j)]=\mathbb{E}_x[|\tilde{Q}_j-\theta_j|]+\mathbb{E}_x[x_j\theta_j].
\end{equation}
By construction, $ \mathbb{E}_x[x_j]=0 $, so
\begin{equation}\label{Qcontri}
\mathbb{E}_x[x_j\tilde{Q}_j]=\mathbb{E}_x[|\tilde{Q}_j-\theta_j|].
\end{equation}
Using Theorem \ref{inequality}, (\ref{Qcontri}) can be lower bounded as
\begin{equation}
\mathbb{E}_x[|\tilde{Q}_j-\theta_j|]\geq \exp(-O(k))\cdot \mathrm{stddev}_x[\tilde{Q}_j-\theta_j].
\end{equation}
As \begin{equation}\label{varineq}
\mathrm{Var}_x[\tilde{Q}_j-\theta_j]\geq \sum_{l\in N_j} \hat{\psi}_l^2\geq \exp(-O(k))\cdot |N_j|,
\end{equation}
one immediately sees that
\begin{equation}
\mathbb{E}_\psi\mathbb{E}_x[x_j\tilde{Q}_j]\geq \exp(-O(k))\cdot \sqrt{|N_j|}.
\end{equation}
	
We would like the algorithm to choose the initial partition $ (F,G) $ uniformly at random. Hence $ \mathbf{P}[i\in G]=1/2 $. Conditioning on $ i $ in $ G $, the probability of a constraint involving $ i $ being active is at least $ 2^{-k+1} $. Hence, conditioned on $ i\in G $,
\begin{equation}
\mathbb{E}[|\mathbf{N}_i|]\geq \exp(-O(k))\cdot \mathrm{deg}(i),
\end{equation}
where $ |\mathbf{N}_j|=A_1+\cdots +A_{\mathrm{deg}(i)} $ is now the sum of the indicator functions for the constraints.
	
Since each indicator function is a product of the indicator functions for the variables, $ |N_j| $ can be regarded as a polynomial of degree at most $ k $. Using Theorem \ref{inequality}, one obtains
\begin{equation}
\mathbf{P}\bigl[|\mathbf{N}_j|\geq \mathbb{E}[|\mathbf{N}_j|]\bigr]\geq \exp(-O(k)).
\end{equation}
	
Hence
\begin{equation}
\mathbb{E}\left[\sqrt{|\mathbf{N}_j|}\right]\geq \exp(-O(k))\cdot \sqrt{\mathrm{deg}(i)}
\end{equation}
	
and we've just shown that 
\begin{equation}
\mathbb{E}_\psi\mathbb{E}_x\left[\sum_l \psi_l-\hat{\psi}_l(\emptyset)\right]\geq \exp(-O(k))\sum_{i=1}^n\sqrt{\mathrm{deg}(i)}\geq \exp(-O(k))\frac{m}{\sqrt{D}}
\end{equation}
\end{proof}

\section{Conclusion}\label{sec:conclusion}
We showed, using the quantum approximate optimization algorithm, that for a CSP of bounded degree $D$, if each constraint associated with a scope can be chosen independently from some probability distribution such that typicality is satisfied, then with high probability, one can produce an assignment satisfying $ \mu+\Omega(1/\sqrt{D}) $ fraction of the constraints. In other words, for typical instances, one can find an assignment satisfying $ \mu+\Omega(1/\sqrt{D}) $ fraction of the constraints. This is completely general, with Max-$ k $XOR and Max-$ k $SAT being the obvious examples. It also does not assume any structure for the underlying constraint hypergraph. It however, \emph{does not mean} we can find an assignment satisfying $ \mu+\Omega(1/\sqrt{D}) $ fraction of the constraints, for \emph{every} instance. That task is possible for Max-$ k $XOR ($ k $ odd), but false for Max-2XOR, for example; the question of when such an assignment exists remains an interesting open question. More uniform patterns are expected to arise in CSPs with bounded degree. It would also be interesting to see if there are arguments distinguishing the CSPs in which $ \mu+\Omega(1/\sqrt{D}) $ fraction of the constraints can be satisfied, from those in which only $ \mu+\Omega(1/D) $ fraction of the constraints can be satisfied.

One interesting aspect would be to see if the classical algorithm can achieve a similar result for CSPs with typicality, and also for CSP instances with "no overlapping constraints".

From the QAOA side, it will be interesting to see if one can use the full Hamiltonian, as there's some numerical evidence that this could improve the constant factor in Theorem \ref{typicality} and \ref{triangle}. Also, it's unclear if there are other merits in considering the truncated Hamiltonian, and the justification for using it.

\section*{Acknowledgements}
We would like to thank Edward Farhi for many valuable discussions and suggestions. YZ would also like to thank David Gosset, Aram Harrow, Ashley Montanaro, Ryan O'Donnell and John Wright for some insightful comments. This work is supported by the ARO under Grant No. W911NF-12-0486, and by the Department of Defense.

\begin{appendices}
\section{Sketched proof of the general case for the quantum algorithm}
For the most general case, we have
\begin{equation}
\psi_l(Z)=\sum_{K\subset K_l}\hat{\psi}(K)Z^S
\end{equation}
where WLOG, we can assume all $ K $ here are proper subsets of $ K_l $.

Find a term with the highest degree. There may be more than one such term, in which case we pick an arbitrary one. Suppose this term is $ \hat{\psi}_l(K_{r_l})Z^{K_{r_l}}=\hat{\psi}_l(K_{r_l})Z_{i_1}\cdots Z_{i_{r_l}} $, where $ r_l<k_l $ and $ K_{r_l}=\{i_1,\dots,i_{r_l}\} $.

Define the truncated Hamiltonian to be
\begin{equation}
C=\sum_l\hat{\psi}_l(K_{r_l})Z^{K_{r_l}}.
\end{equation}

The contribution of one term is
\begin{equation}\label{leadingoneterm}
\mathbb{E}_\psi[\hat{\psi}_u(K_{r_u})\bra{s}e^{i\gamma C}e^{i\beta B}Z^{K_{r_u}}e^{-i\beta B}e^{-i\gamma C}\ket{s}]
\end{equation}

Note that now the highest degree terms from different constraints can be the same. Taking that into account, write
\begin{equation}
C=\tilde{C}+\bar{\psi}_u(K_{r_u})Z^{K_{r_u}}
\end{equation}
where
\begin{equation}
\bar{\psi}_u(K_{r_u})=\hat{\psi}_u(K_{r_u})+\sum_{\{l: K_{r_l}=K_{r_u},l\neq u\}}\hat{\psi}_l(K_{r_l})
\end{equation}
and $ \hat{\psi}_l(K_{r_l}) $'s are independent since they come from different constraints.

With a bit of effort, one can show that
\begin{align}
&\mathbb{E}_\psi [\hat{\psi}_u(K_{r_u})\cos(2\gamma \bar{\psi}_u(K_{r_u}))]=O\left(\frac{1}{D}\right)\\
&\mathbb{E}_\psi [\hat{\psi}_u(K_{r_u})\sin(2\gamma \bar{\psi}_u(K_{r_u}))]\geq\frac{2\mathrm{Var}[\hat{\psi}_u(K_{r_u})]\exp(-2\text{max}\mathrm{Var}[\psi_{1\dots r}])}{\sqrt{D}}+O\left(\frac{1}{D^{3/2}}\right).
\end{align}
By arguments (\ref{odd})-(\ref{kave1}), it can be shown that the truncated part contributes as $ \Theta(1/\sqrt{D}) $.

It can be shown similarly that the remaining terms contribute as $ O(1/\sqrt{D}) $.
\end{appendices}

\end{document}